\documentclass[a4paper]{article}
\usepackage{latexsym}
\usepackage{graphics}
\usepackage{epsfig}
\usepackage{amsmath}
\usepackage{amssymb} 
\usepackage{fancybox} 
\usepackage{bm}
\usepackage{amsthm}
\usepackage{mathrsfs}
\def\A{\mathcal{A}}
\def\B{\mathcal{B}}
\def\CC{\mathcal{C}}

\def\V{{\bm V}}

\def\W{\mathscr{W}}
\def\V{\mathscr{V}}
\def\p{\mathscr{P}}

\def\M{{\cal M}}

\def\X{\mathscr{X}}
\def\K{\mathbb{K}}

\def\R{\mathbb{R}}
\def\C{\mathbb{C}}

\newtheorem{theorem}{\bf Theorem}[section]

\newtheorem{definition}{\bf Definition}[section]

\newtheorem{example}[section]{\bf Example}
\newcommand{\QED}{\nobreak \ifvmode \relax \else
      \ifdim\lastskip<1.5em \hskip-\lastskip
      \hskip1.5em plus0em minus0.5em \fi \nobreak
      \vrule height0.75em width0.5em depth0.25em\fi}

\begin{document}
%%%% Article title to be placed here
\title{
Advantages of probability amplitude over probability density in quantum mechanics
}
\author{Yoshimasa Kurihara$^a$\footnote{yoshimasa.kurihara@kek.jp}~~and Nhi My Uyen Quach$^b$\\
$~$\\
$^a$ {\it The High Energy Accelerator Organization (KEK), 
Tsukuba, Ibaraki 305-0801, Japan}\\ 
$^b$ {\it The Graduate University for Advanced Studies, Tsukuba, Ibaraki 305-0801, Japan}}
\date{}
\maketitle
%%%% Abstract text to be placed here %%%%%%%%%%%%
\begin{abstract}
We discuss reasons why a probability amplitude, which becomes a probability density after squaring, is considered as one of the most basic ingredients of quantum mechanics. First, the Heisenberg/Schr{\" o}dinger equation, an equation of motion in quantum mechanics, describes a time evolution of the probability amplitude rather than of a probability density. There may be reasons why dynamics of a physical system are described by amplitude. In order to investigate one role of the probability amplitude in quantum mechanics, specialized codeword-transfer experiments are designed using classical information theory. Within this context, quantum mechanics based on probability amplitude provides the following: i) a minimum error of the codeword transfer; ii) this error is independent of coding parameters; and iii) nontrivial and nonlocal correlation can be realized. These are considered essential advantages of the probability amplitude over the probability density.
\end{abstract}
%%%%%%%%%%%%%%%%%%%%%%%%%%%
\section{Introduction}
Quantum mechanics (QM) is considered the most basic theory of nature. All phenomena including those of the gravitational force are considered to be expressed by a language of QM. 
However, an essential understanding of the basic nature of QM yet to be realized, and efforts to look for more fundamental explanations continue. Of course, QM itself is a self-consistent theory and requires no fundamental reasoning to support its truths beyond what are gains from experiments. Still, it is worth pursuing more basic reasons which determine QM to be the most fundamental law of nature. For instance, Wheeler asked ``Why the quantum?" and discussed the relation between QM and information theory \cite{citeulike8529862,Wheeler1991fs}. In this report we attempt to answer the same question from Wheeler's point of view.
One of the most essential differences between quantum and classical mechanics is the former's need for a probabilistic treatment of theoretical predictions. One cannot avoid the probabilistic interpretation of a wave function proposed by Born \cite{springerlink:10.1007/BF01397184}, which is now known as the {\it Copenhagen interpretation}. A fundamental equation of QM, the Heisenberg/Schr{\" o}dinger equation, does not describe the behavior of a physical observable nor its probability density; rather, it describes the probability amplitude, which is a characteristic of QM and possesses no classical counterpart. 
(In a narrow sense,``quantum amplitude'' is a complex number whose square of the absolute value is a probability. In this report, we use a word ``quantum amplitude'' not only for complex numbers, but also for vectors whose square of the absolute value is a probability.)
This report considers reasons why fundamental laws of physics are described by probability amplitude instead of probability density, leaving aside the question of why probability itself is necessary. To clarify essential properties of probability amplitude, codeword-transfer experiments are designed on the basis of classical information theory. Taking into account the discussions on these experiments, three essential advantages of probability amplitude over probability density are pointed out in the following sections.

First, definition of quantum system and probability amplitude are given in Section 2 under a very general mathematical framework. Then, codeword-transfer experiments are designed within classical information theory to investigate the role of probability amplitude. Experiments using a stochastic algorithm cannot avoid statistical error due to sample number. In Section 3, we show that a coding method based on probability amplitude should minimize statistical error. Moreover, statistical errors of the codeword-transfer are independent of the parametrization allowing each character to be transferred; this is shown in Section 4. Another essential feature of QM is its lack of local realism, which can be judged by Bell's inequality. This local realism and Bell's inequality are described using terminology of classical information theory, again as the codeword-transfer experiment. A method based on the probability amplitude can induce a violation of Bell's inequality, as shown in Section 5. Throughout this report, classical information theory is used to describe codeword-transfer experiments.
%%%%%%%%%%%%%%%%%%%%%%%%%%%%%%%%%%%%%
\section{General quantum system}
A general framework to define the probability amplitude appearing in QM is considered in this section. Here we emphasis algebraic aspects of QM and ignore dynamical ones. The question which must be asked here is ``What minimum set of assumptions makes a system look like quantum mechanics?'' We propose the following elements as indispensable ingredients for QM.
\begin{definition}\label{gcy}{\bf (Quantum Space)}\\
$\K$ is any field and $V$ is a linear (vector) space on it. $\K$ is named as a base field and is associated to each point of a set, $\M$. State vector and probability measure are introduced on these spaces as follows.
\end{definition}
\begin{enumerate}
\item A map from a point on $\M$ to a tensor product of a vector space $V$, 
\begin{eqnarray}
{\bm \Psi}:\M\rightarrow V^k=\underbrace{V\otimes\cdots\otimes V}_{k}
:x\mapsto{\bm \Psi}(x)=\{\psi^1(x),\cdots,\psi^k(x)\},
\end{eqnarray}
is named state vectors. Here, $\M$ is named the base set and $x$ is a point on it. 
\item A map from the state vector to a real number such as
\begin{eqnarray}
\mu:V\rightarrow\R:\psi^i\rightarrow \mu(\psi^i)\in[0,1],~~~~i=1,\cdots,k
\end{eqnarray}
is named a probability measure. The index $i$ on $\psi^i$ runs from $1$ to $k$. The sequential map 
\begin{eqnarray}\label{mu}
\mu\circ\psi^i:\M\rightarrow V\rightarrow\R:x\mapsto \mu^i(x)=\mu(\psi^i)(x)
\end{eqnarray}
is also called a probability measure and represented by the same symbol, $\mu$, when $V$ are obvious. 
\item The probability measure must be normalized as
\begin{eqnarray}
\int_{x\in\Gamma\subseteq\M}\mu(\psi^i)(x)&=&1
\end{eqnarray}
for each $i$, where $\Gamma$ is an appropriate subset of the base set $\M$. Since the probability measure is considered as Lebesgue measure, the integral should be interpreted as summation when $\M$ is a discrete set.
\item The set $\{\K, V, {\bm \Psi}, \mu\}$ is named a ``quantum space.''
$\QED$
\end{enumerate}
To construct QM, these conditions are necessary, but are not sufficient.
For standard relativistic QM (or quantum field theory), we take Hilbert space as a vector-space $V$ on a field of complex numbers $\C$. State vector can be constructed using square integrable functions on a given support. The state vector is associated with each point of the Minkowski manifold as a base set. (Sometimes a Fourier transformation of $\psi^i$ defined in the momentum manifold is used instead of $\psi^i$ itself. In that case, a corresponding Hilbert space is called ``Fock space''.) The probability measure is introduces as $\mu(\psi^i)=|\psi^i|^2$. For the normalization, $\Gamma$ is taken as a hyper-surface on $\M$ such that any two points on $\Gamma$ have a space-like distance each other. (Or it is normalized in the momentum space.) When the probability measure is defined as square of the absolute value of the state vector, the state vector is called a ``{\it probability amplitude}'' in this report, hereafter. In this report, simple quantum spaces are used since only algebraic aspects of QM are of interest here.
%%%%%%%%%%%%%%%%%%%%%%%%%%%%%%%%%%%%%
\section{Minimization of measurement error}
First, let us consider a statistical error for measurements of a single physical observable on the quantum space defined in the previous section. A codeword-transfer experiment simulating standard QM in a much simpler quantum space, retaining essential properties, is introduced here. In information theory, an encoding method which minimizes statistical error among methods using stochastic algorithms is known. The method using probability amplitude is shown to be an example of such an encoding method giving minimum errors. Terminology of classical information theory used here can be found in {\bf Appendix \ref{ap1}} and references\cite{covert91-12-11,kurihara}.
\begin{definition}\label{codetf}{\bf (Stochastic codeword-transfer experiment)}\\
The experiment satisfying the following conditions is called a {\it stochastic codeword-transfer experiment}:
\end{definition}
\begin{enumerate}
\item Alice ($\A$) transfers a set of $m$ different codewords $\W=\{w_1,\cdots,w_m\}$ to Bob ($\B$) after converting them to state vectors $\psi\in\V$, where $\V$ is a $m$-dimensional vector space.
\item $\B$ receives a state vector sent from $\A$ and obtained one of codewords $\W$ by measuring them. Here meaning of ``measuring'' will be explain in following items. 
\item The same probabilistic function of
\begin{eqnarray*}
\mu_{c}:\V\rightarrow\R:\mu_{c}(\psi)(\omega_i)\mapsto p_i\in[0,1]
\end{eqnarray*}
is given for $\B$. The value $p_i$ gives a probability to observe a codeword $\omega_i$. Only one vector space $\V$ appears here, then the function $\mu_{c}(\psi)(\omega_i)$ will be written as $\mu_{c}(\omega_i)$, hereafter.
\item Probabilistic function $\mu_{c}$ is normalized as:
\begin{eqnarray*}
\sum_{i=1}^m\mu_{c}(\omega_i)&=&1.
\end{eqnarray*}
\item $\A$ can repeat to send a finite number ($n$ times here) of the same state vectors to $\B$.
\item $\B$ obtains $n$ independent codewords by measuring sets of state vectors sent from $\A$, such as $\X=\{x_1,\cdots,x_n\}$.
\item $\B$ has an unbiased estimator to obtain a set of real numbers ${\bar x}_i\in[0,1]$ from measured data as
\begin{eqnarray*}
{\bar x}_i=T_i(\X)&=&\frac{1}{n}\sum_{j=1}^n\Theta_i(x_j),\\
\Theta_i(x_j)&=&
\begin{cases}
    1     & (x_j=\omega_i), \\
    0     & (x_j\neq\omega_i).
  \end{cases}
\end{eqnarray*}
Here ${\bar x}_i$ converges in probability to $p_i$ when $n\rightarrow\infty$, thanks to the law of large numbers.
\item Finally $\B$ obtains a sequence of numbers $\{{\bar x}_1,\cdots,{\bar x}_m\}$, which $\A$ intended to send.
$\QED$
\end{enumerate}
This codeword-transfer experiment is constructed on the quantum space $\{\W, \V, \psi, \mu_c\}$ as defined above. In this case, positions where $\A$ or $\B$ exists are not specified. No dynamical structure is assumed to transport a state vector from $\A$ to $\B$ here, however it is just assumed that these two points are separated from each other and there is no way to communicate other than the state-vector transfer.
A question to ask here is how may one find the probability measure $\mu_c$, which maps the state vector $\psi$ to a real number $\mu_c(\psi)$ to minimize an error of this experiment for any $\psi$. The answer is already known as a theorem, which was first obtained by Fisher \cite{Fisher}. Wootters stated this theorem \cite{Wootters1} without any proof but later provided the same by introducing a statistical distance \cite{Wootters2}. Recently Wootters discussed this subject again in \cite{e15083220}. Here we state the theorem clearly again and give an independent and much simpler proof using an information theory.
%%
%% Fisher--Wootters
%%
\begin{figure}[tb]
  \begin{center}
     \includegraphics[height=5cm]{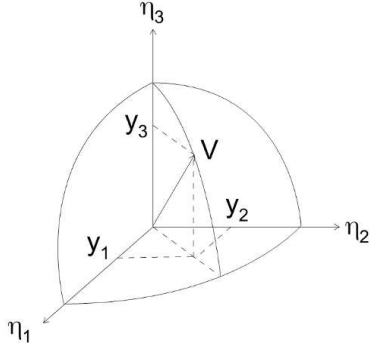}%
    \caption{\footnotesize Example for $m=3$: assignment of a vector $\V$ on two-dimensional sphere.}
    \label{theor1}
  \end{center}
\end{figure}
\begin{theorem}\label{wootters}{\bf (Fisher--Wootters)}\\
Among stochastic codeword-transfer experiments, that which employs the following probability measure gives the smallest error to measure a single codeword from a set of codewords:
\end{theorem}
\begin{enumerate}
\item $\A$ selects a set of codewords $\W=\{w_1,\cdots,w_m\}$ and a sequence of numbers $\p=\{p_1,\cdots,p_n\}$ which is intended to be sent to $\B$. The $\p$ is normalized as $\sum_{i=1}^{n}p_i=1$.
\item $\A$ prepares an $m$-dimensional Euclid space $\R^m$ and orthonormal bases $\{\eta_1,\cdots,\eta_m\}$.\label{wootters12}
\item $\A$ sets a state vector $\psi$ as to map each element of $\p$ at a point on a unit sphere $S^{m-1}_1$ centered at the origin of $~\R^m$ to be $p_i=|y_i|^2$, where $y_i$ is an $i$'s component of the position of $\psi$ on $S^{m-1}_1$ by the orthonormal bases defined above. 
\item $\B$ obtains a codeword $\omega_i$ with the probability measure $\mu_c(\omega_i)=|y_i|^2$.
$\QED$
\end{enumerate}
\begin{proof}
\begin{enumerate}
\item {\bf The smallest error} $\Rightarrow~\mu_c(\omega_i)=|y_i|^2${\bf :}\\
Data after $n$ independent measurements are expressed as $\X=(x_1,\cdots,x_n)$ with the probability $\mu_c(\omega_i)=|y_i|^2$. The probability density to obtain a set of data $\X$ is assumed to be expressed as ${\rm f}(\X;\psi)=\mu_c(\omega_i)$, where $\mu_c(\omega_i)$ used defined as an equation (\ref{mu}). Then the Fisher information matrix (FIM) \cite{covert91-12-11} can be written as
\begin{eqnarray*}
J_{ij}&=&{\rm f}(\X;\psi)
\frac{\partial\log{{\rm f}(\X;\psi)}}{\partial \omega_i}
\frac{\partial\log{{\rm f}(\X;\psi)}}{\partial \omega_j}\\
&=&\sum_{k=1}^{m}\mu_c(\omega_k)
\frac{\partial \mu_c(\omega_k)/\partial \omega_i}{\mu_c(\omega_k)}
\frac{\partial \mu_c(\omega_k)/\partial \omega_j}{\mu_c(\omega_k)}.
\end{eqnarray*}
The functions $\mu_c(\omega_i)$ are not independent of each other owing to conservation of the total probability, $\sum_{i=1}^m\mu_c(\omega_i)=1$. We can assume that all $\mu_c(\omega_i)~(i\geq 2)$ are independent except $\mu_c(\omega_1)=1-\sum_{j=2}^m\mu_c(\omega_j)$ without any loss of generality. Since all other $\mu_c(\omega_{i\neq1})$, except this correlation due to the conservation of probability, can be set to be independent after appropriate linear transformation of $\mu_c$, the FIM can be taken to be a diagonal matrix.
Here we use a short-hand expression, $\mu_c(\omega_i)=\mu_{i}$, $d\mu_c(\omega_j)/d\omega_i=\mu_{j,i}$, and $\sum_{j=2}^m\mu_c(\omega_j)=\bar{\mu}$; then the diagonal components of the FIM can be written as
\begin{eqnarray*}
J_{ii}&=&\sum_{k=1}^m\mu_k\left(\frac{\mu_{k,i}}{\mu_k}\right)^2\\
&=&\mu_1\left(\frac{\mu_{1,i}}{\mu_1}\right)^2
+\sum_{k=2}^m\mu_k\left(\frac{\mu_{k,i}}{\mu_k}\right)^2\\
&=&(1-\bar{\mu})
\left(\frac{\partial(1-\bar{\mu})/\partial\omega_i}{1-\bar{\mu}}\right)^2
+\mu_i\left(\frac{\mu_{i,i}}{\mu_i}\right)^2\\
&=&
\frac{\mu_{i,i}^2}{1-\bar{\mu}}
+\frac{\mu_{i,i}^2}{\mu_i}.
\end{eqnarray*}
Here the independence of all $\mu_{k\geq2}$ each other is used second line to third line in above calculations.
The minimum value of $J_{ii}$ is obtained when $\bar{\mu}=\mu_i$ within the allowed region of $\mu_i\leq\bar{\mu}\leq 1$. Then we get
\begin{eqnarray*}
\min\{J_{ii}\}&=&\tilde{J}_{ii}
=\frac{\mu_{i,i}^2}{\mu_i(1-\mu_i)}.
\end{eqnarray*}
On the other hand, measured data after $n$ independent measurements must follow a multinomial distribution, whose covariance matrix $\bm{\sigma}$ is
\begin{eqnarray*}
\sigma_{ij}=
  \begin{cases}
    n{\tilde p}_i{\tilde p}_j     & (i \neq j), \\
    n{\tilde p}_i(1-{\tilde p}_i) & (i = j),  
  \end{cases}
\end{eqnarray*}
where $\tilde{p}_i$ is measured probability of an $i$th codeword.
Then, after $n$ independent measurements through estimator $T$ defined in {\bf Definition~2.6}, a covariant matrix ${\bm \Sigma}(\X)$ can be expressed as
\begin{eqnarray*}
\Sigma_{ij}(\X)&=&\frac{1}{n}\sum_{k=1}^n 
(\Theta_i(x_k)-\bar{x}_i)
(\Theta_j(x_k)-\bar{x}_j)\\
&\approx&\frac{1}{n}\sigma_{ij}.
\end{eqnarray*}
Then, diagonal components of the covariant matrix become
\begin{eqnarray*}
\Sigma_{ii}(\X)&\approx&\frac{1}{n}\sigma_{ii}\\
&=&{\tilde p}_i(1-{\tilde p}_i).
\end{eqnarray*}
In general, measured probability ($\tilde{p}_i$) differs from true probability ($\mu_i$); however, it is certain that the error of $|{\tilde p}_i-\mu_i |$ will be less than any small value after a sufficient number of events accumulates, as a result of the law of large numbers and the assumption that the estimator is unbiased. Then, we use $\mu_i$ instead of $\tilde{p}_i$ in the discussions that follow.
The probability $\mu_i$ that maximizes diagonal components of the covariant matrix is given as $\mu_i=1/2$ due to $d\Sigma_{ii}/d\mu_i=(1-2\mu_i)=0$. Then, the diagonal components of the covariant matrix are given as $\Sigma_{ii}=1/4$. The Cram\'{e}r--Rao inequality \cite{Rao,Cramer,covert91-12-11} gives the lower bound of the covariant matrix as
\begin{eqnarray*}
{\bm \Sigma}(\bm \mu)\geq{\bm J}^{-1}.\label{c-r}
\end{eqnarray*}
A possible range of the inverse of the FIM is
\begin{eqnarray*}
\frac{\mu_i(1-\mu_i)}{\mu_{i,i}^2}\geq\left({\bm J}^{-1}\right)_{i,i}
\geq\frac{1}{{\tilde J}_{i,i}}
\geq 0,
\end{eqnarray*}
where we use the FIM (${\bm J}$) is a diagonal matrix.
Then a solution of the following differential equation gives the minimum variance in general:
\begin{eqnarray*}
\frac{\mu_{i,i}^2}{\mu_i(1-\mu_i)}&=&4\\
\Rightarrow \mu_{i,i}^2
&=&4\mu_i\left(1-\mu_i\right).
\end{eqnarray*}
The solution of this equation can be obtained as
\begin{eqnarray*}
\mu_i&=&\cos^2{(\omega_i+\phi_i)},
\end{eqnarray*}
where $\phi_i$ is an arbitrary phase factor. This phase factor corresponds to a rotation of the coordinate system prepared in {\bf Theorem~\ref{wootters}} and gives no essential effect on the result. Then we set $\phi_i=0$ hereafter as $\mu_i=\cos^2{\omega_i}$. Each $\omega_i$ gives the same differential equation; then parametrization $y_i=\sqrt{\mu_i}=\cos{\omega_i}$ gives the lowest value of the variance, which is nothing other than the direction cosine of the vector $V$, whose endpoint is on the unit sphere $S^{m-1}_1$. Then, the method to give the minimum variance is: {\bf i)} normalize the codeword $\omega_i$ to $0\le\omega_i\le\pi/2$; {\bf ii)} map on the $S^{m-1}_1$ as $\omega_i$ to be an angle from axis $\eta_i$; set {\bf iii)} the probability to observe the codeword $\omega_i$ to be $\cos^2{\omega_i}$, which are the same as the assumptions of the theorem.
\item $\mu_c(\omega_i)=|y_i|^2~\Rightarrow$ {\bf the smallest error:}\\
When we set $\mu_i=|y_i|^2=\cos^2\omega_i$, the diagonal components of a covariant matrix become
\begin{eqnarray*}
\Sigma_{ii}(\bm{\mu})&=&\mu_i(1-\mu_i)\\
&=&\cos^2\omega_i(1-\cos^2\omega_i)\\
&=&\cos^2\omega_i\sin^2\omega_i.
\end{eqnarray*}
Then the minimum value of $\Sigma_{ii}$ is obtained to be $1/4$ at $\omega_i=\pi/4$. On the other hand, the diagonal component of the FIM matrix can be
\begin{eqnarray*}
\tilde{J}_{ii}&=&\Bigl|\mu_{i,i}\Bigr|^2
\frac{1}{\mu_i(1-\mu_i)}\\
&=&4\cos^2\omega_i\sin^2\omega_i/(\cos^2\omega_i\sin^2\omega_i)\\
&=&4.
\end{eqnarray*}
Then $\tilde{J}_{ii}^{-1}=1/4$, which matches the minimum value of $\Sigma_{ii}$.
\end{enumerate}
\end{proof}
In the above decoding method, a relation between probability amplitude and density is algebraically the same as in the standard QM, which means the latter employs a coding method that minimizes statistical error among other stochastic methods. This is our first example outlining the advantage of the method using probability amplitude. 
%%%%%%%%%%%%%%%%%%%%%%%%%%%%%%%%%%%%%
\begin{figure}
  \begin{center}
     \includegraphics[height=3cm]{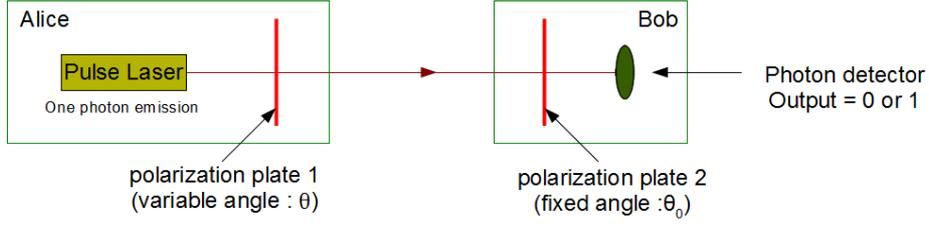}
    \caption{\footnotesize Code-transfer experiment realized by using polarized laser beam.}
    \label{exp1}
  \end{center}
\end{figure}
%%%%%%%%%%%%%%%%%%%%%%%%%%%%%%%%%%%%%
\section{Parametrization independence of a measurement error}
Related to the {\bf Theorem~\ref{wootters}}, one can prove following theorem, which is also given by Wootters \cite{Wootters1,Wootters2} and is important to consider one role of the probability amplitude.
\begin{theorem}\label{wootters2}~\\
The encoding rule given by the Fisher--Wootters theorem gives uniform errors independent of its parametrization.$\QED$
\end{theorem}
\begin{proof}
A set of state vectors $\V=\{\omega_1,\cdots,\omega_m\}$ are encoded as $y_i=\cos{\omega_i}$, $\sum_{i=1}^my_i^2=1$, according to the Fisher--Wootters theorem.
Looking at an $i$th element $\omega_i$, one sees a relation between an error of estimation $\delta\omega_i$ and an error to measure the parameter $\delta y_i$ as 
$
\left|\delta y_i/\delta\omega_i\right|=\sin{\omega_i}.
$
Under the normalization condition of
$
\sum_{i=1}^my_i^2=\sum_{i=1}^m\cos^2{\omega_i}=1,
$
the total error after measuring $n$ independent data becomes
\begin{eqnarray*}
\sigma^2&=&\frac{1}{n}
\sum_{i=1}^m\left|\frac{\delta y_i}{\delta\omega_i}\right|^2\\
&=&\frac{1}{n}\sum_{i=1}^m\left|\sin^2{\omega_i}\right|\\
&=&\frac{1}{n}\sum_{i=1}^m\left|1-\cos^2{\omega_i}\right|\\
&=&\frac{m-1}{n}\\
\Rightarrow\sigma&=&\sqrt{\frac{m-1}{n}},
\end{eqnarray*}
which means a mean-square error is determined by the statistics per degree of freedom and independent of the position on an $m$-dimensional sphere. A factor $\sigma^2\propto 1/n$ follows from the central limit theorem.
\end{proof}
\begin{example}\label{leaser}{\bf (Codeword-transfer experiment realized using a polarized laser beam)}\end{example}
\vskip -0.4cm
\noindent
Let us consider the codeword-transfer experiment defined in {\bf Definition~\ref{codetf}} for a realistic quantum system: the pulse laser has a polarizer ($\lambda/4$ plate). (See Fig. \ref{exp1}.) Alice ($\A$) has experimental equipment consisting of a pulse laser and a polarizer and can transfer a single photon with linear polarization with any polarization plane to Bob ($\B$).  $\B$ has a $\lambda/4$ plate with fixed plane and photon detector with $100\%$ efficiency. $\A$ knows the angle of the polariser plane of $\B$, say $\theta_0$, and has a clock exactly synchronised to that of $\B$. $\A$ assigns codewords on equally separated points on a unit circle, and selects an integer, say $j$. Then $\A$ sets an angle of the polariser according to a codeword to be $\theta=\theta_0+\alpha$, where $\alpha=j/2\pi$. $\A$ transfers one photon a second and $n$ photons in total. $\B$ measures photons behind the $\lambda/4$ plate. If $\B$ observes a photon, he records ``$1$'' and if not, he records ``$0$''. As a result $\B$ obtains data $\X_n=\{X_1,X_2,\cdots,X_n\}=\{1,1,0,1,0,\cdots\}$, and decodes them to one real number with average $\bar{x}=\sum_i^n X_i/n$. According to quantum mechanics this number must be $\bar{x}=\sin{\alpha}$. Finally, $\B$ obtains a number which $\A$ intended to send. This codeword-transfer experiment satisfies {\bf Definition~\ref{codetf}}, which means quantum mechanics gives codeword-transfer experiments with the smallest errors, given by Theorem.~$\ref{wootters}$.
$\QED$
%%%%%%%%%%%%%%%%%%%%%%%%%%%%%%
\section{Nonlocal realism}
A point definitely distinguishing QM from classical mechanics is that QM does not have local realism. Related to this fact, there are two important theorems: violation of Bell's inequality\cite{Bell} and Kochen--Specker theorem\cite{kochen1975problem}. Both theorems are related to a correlation of two independent measurements. It is shown in this section that these two theorem can be realized again using the probability amplitude. In order to discuss a correlation of two independent measurements, a double codeword-transfer experiment is designed.
\begin{definition}\label{dcodetf}{\bf (Stochastic double codeword-transfer experiment)}
A {\it stochastic double codeword-transfer experiment} is defined by extending {\bf Definition \ref{codetf}} as follows:
\end{definition}
\begin{enumerate}
\item\label{dcodetf-i1} Alice $(\A)$ transfers two sets of $m$ different codewords and state vectors, $\W_\alpha=\{\alpha_1,\cdots,\alpha_m\}$ and $\W_\beta=\{\beta_1,\cdots,\beta_m\}$, to Bob $(\B)$ and Charley $(\CC)$ after converting them to state vectors $\psi_\alpha\in\V_\alpha$ and $\psi_\beta\in\V_\beta$, where $\V_\alpha$ and $\V_\beta$ are $m$-dimensional vector spaces.
\item\label{dcodetf-i2} $(\B)$ and $(\CC)$ are placed opposite to $\A$ and receive state vectors sent from $\A$, stochastically choose one of the two sets to be measured. Neither $\B$ and $\CC$  know which set is chosen by the other $($independence of set selection$)$.
\item Encoding is performed using the following probabilistic function:
\begin{eqnarray*}
\mu_{dc}:(\V_\alpha\oplus \V_\beta)\otimes(\V_\alpha\oplus \V_\beta)\rightarrow\R:(\gamma_i,\gamma_j)\mapsto\mu_{dc}(\gamma_i,\gamma_j)=p_{i,j}\in[0,1], 
\end{eqnarray*}
where $\V_\alpha\oplus \V_\beta=(\alpha_1,\cdots,\alpha_m,\beta_1,\cdots,\beta_m)=(\gamma_1,\cdots,\gamma_{2m})$ and $1\le i,j\le 2m$. First (second) slot of $\mu_{dc}$ are for state vectors sent to $\B$ ($\CC$), respectively.
\item $\B$ and $\CC$ select for measurement one of the state vectors $\psi_A$ or $\psi_B$, independently. Possible combinations of measured codewords are $\{(\alpha_k,\alpha_l)$, $(\alpha_k,\beta_l)$, $(\beta_k,\alpha_l)$, $(\beta_k,\beta_l)\}$. Probabilistic function $\mu_{dc}$ is normalized as:
\begin{eqnarray*}
~^\forall\gamma_j,~~\sum_{i=1}^m\mu_{dc}(\alpha_i,\gamma_j)=1,&~&
~^\forall\gamma_j,~~\sum_{i=1}^m\mu_{dc}(\beta_i,\gamma_j)=1,\\
~^\forall\gamma_i,~~\sum_{j=1}^m\mu_{dc}(\gamma_i,\alpha_j)=1,&~&
~^\forall\gamma_i,~~\sum_{j=1}^m\mu_{dc}(\gamma_i,\beta_j)=1.
\end{eqnarray*}
However it does not guarantee that all the probability measures, $\mu_{dc}(\alpha_k,\alpha_l)$, $\mu_{dc}(\alpha_k,\beta_l)$, $\mu_{dc}(\beta_k,\alpha_l)$, and $\mu_{dc}(\beta_k,\beta_l)$, exist at the same time.
\item $\A$ can send a finite number ($n$ times here) of the same set of state vectors to $\B$ and $\CC$.
\item Measurements:
\begin{enumerate}
\item $\B$ obtains $n$ independent codewords by measuring sets of state vectors sent from $\A$, such as $\X^\B=\{x^\B_1,\cdots,x^\B_n\}$, where $x^\B_i\in\W_\alpha\oplus\W_\beta$.
\item For $\CC$, the same as (\it{a}) with a replacement $\B\rightarrow\CC$. 
\end{enumerate}
\item Estimator:
\begin{enumerate}
\item $\B$ has an unbiased estimator to obtain a set of real numbers ${\bar x}_i\in[0,1]$ from measured data as
\begin{eqnarray*}
{\bar x}^\B_i&=&T_{i}(\X^\B)\\
&=&\frac{\sum_{k=1}^n\Theta_i(x^\B_k)}
{\sum_{l=1}^{2m}\sum_{k=1}^n\Theta_l(x^\B_k)}\\
\Theta_i(x_j)&=&
\begin{cases}
    1     & (x_j=\gamma_i), \\
    0     & (x_j\neq\gamma_i),
\end{cases}
\end{eqnarray*}
where $i$ runs from $1$ to $2m$.
\item  For $\CC$, the same as above with a replacement $\B\rightarrow\CC$.

\end{enumerate}
\item After completing measurement, $\B$ and $\CC$ make a table 
$
{\bar x}_{i,j}=({\bar x}^\B_i,{\bar x}^\CC_j),
$
where ${\bar x}_{i,j}$ converges in probability to $p_{i,j}$ when $n\rightarrow\infty$, thanks to the law of large numbers.$\QED$
\end{enumerate}
%%%%%%%%%%%%%%%%%%%%%%%%%%%%%%%%%%%%%
\noindent
Bell's inequality is a critical test to distinguish a nonlocal theory from a local one. This theorem can be expressed by the language of classical information theory \cite{Braunstein}. We state this theorem and give a proof in the context of {\bf Definition \ref{dcodetf}}.
\begin{theorem}\label{bell}{\bf (Bell)}\\
Let us consider a case with a complete table to give the probability of observing any pair of codewords as 
\begin{eqnarray*}
P(\alpha_{i_1},\alpha_{i_2},\beta_{j_1},\beta_{j_2})&=&
\mu_{dc}(\alpha_{i_1},\alpha_{i_2})
\mu_{dc}(\beta_{j_1},\beta_{j_2})\\&+&
\mu_{dc}(\alpha_{i_1},\beta_{j_2})
\mu_{dc}(\beta_{j_1},\alpha_{i_2})\\&+&
\mu_{dc}(\beta_{j_1},\alpha_{i_2})
\mu_{dc}(\alpha_{i_1},\beta_{j_2})\\&+&
\mu_{dc}(\beta_{j_1},\beta_{j_2})
\mu_{dc}(\alpha_{i_1},\alpha_{i_2}).
\end{eqnarray*}
These measurements are performed as the stochastic double codeword-transfer experiment defined above. In this case, a conditional entropy follows the inequality
\begin{eqnarray*}
S(\alpha_{i_1}|\alpha_{i_2})\leq 
S(\alpha_{i_1}|\beta_{j_2})+S(\beta_{j_2}|\beta_{j_1})+S(\beta_{j_1}|\alpha_{i_2}).
\end{eqnarray*}
Definitions and necessary formulae for following proof can be found in \cite{covert91-12-11} and summarized in {\bf Appendix \ref{ap2}}.
\end{theorem}
\begin{proof}
On the assumption there exists a complete probability table, $P(\alpha_{i_1},\alpha_{i_2},\beta_{j_1},\beta_{j_2})$, a joint entropy can be written as
%\begin{widetext}
\begin{eqnarray*}
S(\alpha_{i_1},\alpha_{i_2},\beta_{j_1},\beta_{j_2})
&=&-\sum_{i_1,i_2,j_1,j_2}
P(\alpha_{i_1},\alpha_{i_2},\beta_{j_1},\beta_{j_2})
\log{P(\alpha_{i_1},\alpha_{i_2},\beta_{j_1},\beta_{j_2})}\\
&=&S(\alpha_{i_2}\cap\beta_{j_1}\cap\beta_{j_2},\alpha_{i_1})\\
&=&S(\alpha_{i_2}\cap\alpha_{i_1},\beta_{j_1}\cap\beta_{j_2}).
\end{eqnarray*}
%\end{widetext}
Using the chain rule of entropy sequentially, one can get
%\begin{widetext}
\begin{eqnarray*}
S(\alpha_{i_2}\cap\beta_{j_1}\cap\beta_{j_2},\alpha_{i_1})
&=&
S(\alpha_{i_1}|\alpha_{i_2},\beta_{j_1},\beta_{j_2})
+S(\alpha_{i_2}\cap\beta_{j_1},\beta_{j_2})\\
&=&
S(\alpha_{i_1}|\alpha_{i_2},\beta_{j_1},\beta_{j_2})
+S(\beta_{j_2}|\alpha_{i_2}\cap\beta_{j_1})
+S(\alpha_{i_2},\beta_{j_1})\\
&=&
S(\alpha_{i_1}|\alpha_{i_2},\beta_{j_1},\beta_{j_2})
+S(\beta_{j_2}|\alpha_{i_2},\beta_{j_1})
+S(\beta_{j_1}|\alpha_{i_2})+S(\alpha_{i_2}).
\end{eqnarray*}
%\end{widetext}
On the other hand, this joined entropy satisfies
\begin{eqnarray*}
S(\alpha_{i_2}\cap\alpha_{i_1},\beta_{j_1}\cap\beta_{j_2})
&=&S(\alpha_{i_2},\alpha_{i_1})
+S(\beta_{j_1},\beta_{j_2}|\alpha_{i_2},\alpha_{i_1}),\\
&\geq&S(\alpha_{i_2},\alpha_{i_1})\\
&=&S(\alpha_{i_2})+S(\alpha_{i_1}|\alpha_{i_2}).
\end{eqnarray*}
Inequality follows from nonnegativity of entropy. From the property of the probability measure in the probability space,
\begin{eqnarray*}
^\forall \alpha_1, ^\forall \alpha_2\in \V_A,~P(\alpha_1\cap\alpha_2)\leq P(\alpha_1),
\end{eqnarray*}
and the definition of joint entropy, the inequalities
\begin{eqnarray*}
S(\alpha_{i_1}|\alpha_{i_2},\beta_{j_1},\beta_{j_2})
&\leq&S(\alpha_{i_1}|\beta_{j_2}),\\
S(\beta_{j_2}|\alpha_{i_2},\beta_{j_1})
&\leq&S(\beta_{j_2}|\beta_{j_1}),
\end{eqnarray*}
follow. Then Bell's inequality is proved.
\end{proof}
The necessary condition for Bell's inequality, the existence of the complete probability table $P(\alpha_{i_1},\alpha_{i_2},\beta_{j_1},\beta_{j_2})$, corresponds to local realism in the physical terminology. Here, we give an example where Bell's inequality is not maintained.
%
% without a complete probability table
% 
\begin{definition}\label{dcodetf3}{\bf (Stochastic double codeword-transfer experiment without a complete probability table)}\\
Here, the number of codewords in the set is $m=2$~~for simplicity.
\end{definition}
\begin{enumerate}
\item\label{dcodetf3-i1} Set $m=2$ in {\bf Definition~\ref{dcodetf}-\ref{dcodetf-i1}} for two sets of codewords such as
\begin{eqnarray*}
\W_A&=&\{\alpha_1,\alpha_2\},\\
\W_B&=&\{\beta_1,\beta_2\},\\
\W_A\otimes\W_B&=&\{\alpha_1,\alpha_2,\beta_1,\beta_2\}=\{\gamma_1,\gamma_2,\gamma_3,\gamma_4\}
\end{eqnarray*}
and for state vectors as
\begin{eqnarray*}
\V_A\ni\psi_A(\theta_\alpha)&=&(\cos{\theta_\alpha},\sin{\theta_\alpha}),\\
\V_B\ni\psi_B(\theta_\beta)&=&(\cos{\theta_\beta},\sin{\theta_\beta}),
\end{eqnarray*}
where $0\le\theta_\alpha,\theta_\beta\le \pi$. This parametrization configures an example of {\bf Theorem~\ref{wootters}}.
\item The same as {\bf Definition~\ref{dcodetf}-\ref{dcodetf-i2}}.
\item Encoding is performed using following probabilistic function:
\begin{eqnarray*}
\mu_{dc}(\gamma_i,\gamma_j)=\left|\gamma_i\right|^2\left|\gamma_j\right|^2.
\end{eqnarray*}
\item $\B$ and $\CC$ select for measurement one of the elements (codewords) in $\W_A$ or $\W_B$, independently. Before measurement, $\B~(\CC)$ rotates a detector angle up to $\theta_b$ ($\theta_c$). Neither knows the rotating angle of the other. $\B$ and $\CC$ correct this rotation angle after completing all measurements. This rotation does not affect the error of the measurement, owing to {\bf Theorem~\ref{wootters2}}.
\begin{enumerate}
\item If state vectors $\{\alpha_i\}$ and $\{\beta_i\}$ exist locally before the measurement for $\B$, the probability that $\B$ may obtain each codeword can be obtained after rotation as 
\begin{eqnarray*}
\psi_\gamma(\theta_\gamma)\rightarrow
\psi_\gamma(\theta_\gamma-\theta_b)=R(\theta_b)\psi_\gamma, 
\end{eqnarray*}
where $R(\theta)$ is a rotation matrix, $\psi_\gamma\in\V_\alpha\oplus\V_\beta$, and $\theta_\gamma=\theta_\alpha$ or $\theta_\beta$ depending on $\psi_\gamma$.
The probability for $\CC$ is similar to the above. In this case we do not observe any violation of Bell's inequality since we can prepare the complete probability table.
\item\label{dcodetf3-i4b} Suppose the angles $\theta_\alpha$ and $\theta_\beta$ are not fixed before measurement and are fixed when $\B$ or $\CC$ measure the code from $\W_A$ or $\W_B$ and the probability measure $\mu_{dc}$ depends on the result of their decision. Moreover we require that the probability measure does not follow the {\it functional composition condition} (FUNC)\cite{flori2013first}. In a context of the report, the FUNC is a requirement for any function $f$ as arithmetic operations on vectors and real numbers as
\begin{eqnarray*}
~\forall f,~~\mu_{dc}(f(\psi_i,\psi_j),\psi_k)=f\left(\mu_{dc}(\psi_i,\psi_k),\mu_{dc}(\psi_j,\psi_k)\right).
\end{eqnarray*}
A function $f$ in l.h.s. maps real numbers to a real number. On the other hand, $f$ in r.h.s from vectors to a real number. Here we consider a natural isomorphism between real numbers and vectors in operations of addition, subtraction, and (scalar) product, and represented the same symbol $f$. For a current example, the probability measure does not satisfy the FUNC, for example, as
\begin{eqnarray*}
\mu_{dc}(\alpha_1+\alpha_2,\psi_k)&=&
\left|\alpha_1+\alpha_2\right|^2\left|\psi_k\right|^2,\\
\mu_{dc}(\alpha_1,\psi_k)+\mu_{dc}(\alpha_2,\psi_k)
&=&\left(\left|\alpha_1\right|^2+\left|\alpha_2\right|^2\right)\left|\psi_k\right|^2,\\
&\neq&\mu_{dc}(\alpha_1+\alpha_2,\psi_k).
\end{eqnarray*}
Suppose $\CC$ obtains $\alpha_1$~($\alpha_2$). The angle $\theta_\alpha$ for $\B$ is fixed as $\theta_\alpha=\theta_c$~($\theta_\alpha=\pi/2+\theta_c$), i.e., the probability table is now situation-dependent. The state vectors for $\B$ are now 
\begin{eqnarray*}
\psi_A&=&
\begin{cases}
\psi_A(\theta_c-\theta_b)
&\CC~{\rm obtained}~\alpha_1\\
\psi_A(\theta_c+1/2-\theta_b)
&\CC~{\rm obtained}~\alpha_2,\\
\psi_A(\theta_\alpha-\theta_b)
&\CC~{\rm obtained}~\beta_i.
\end{cases}
\end{eqnarray*}
If $\B$ decided to measure a codeword from a set $\W_\beta$, nothing would happen. On the other hand, if $\B$ decided to measure a codeword from the same set as $\CC$, then
\begin{eqnarray*}
&~&\mu_{dc}(\alpha_1,\alpha_1+\alpha_2)\\
&=&\Bigl|
\cos{(\theta_c-\theta_b)}\cos{(\theta_\alpha-\theta_c)}
-\sin{(\theta_c-\theta_b)}\sin{(\theta_\alpha-\theta_c)}\Bigr|^2\\
&=&\cos^2{(\theta_\alpha-\theta_b)},\\
&~&\mu_{dc}(\alpha_2,\alpha_1+\alpha_2)\\
&=&\Bigl|\sin{(\theta_c-\theta_b)}\cos{(\theta_\alpha-\theta_c)}
+\cos{(\theta_c-\theta_b)}\sin{(\theta_\alpha-\theta_c)}\Bigr|^2\\
&=&\sin^2{(\theta_\alpha-\theta_b)}.
\end{eqnarray*}
Again the probability to obtain one of the $\alpha$ can be calculated using only local parameters on $\B$. In both cases, $\B$ can obtain a set of codewords that $\A$ intended to send. The probability table is situation-dependent and there is a possibility that Bell's inequality will be violated.
\end{enumerate}
\item 6.~The same as the {\bf Definition~\ref{dcodetf}}.$\QED$
\end{enumerate}
It is proved that the Kochen--Specker theorem is incompatible the FUNC\cite{flori2013first}. Above stochastic double codeword-transfer experiment is a model of the QM violating the FUNC to incorporate the Kochen--Specker theorem. In order to confirm a violation of Bell's inequality, it is tested numerically according to the above example. A correlation between measured codewords independently obtained by $\B$ and $\CC$ is defined mimically like CHSH\cite{PhysRevLett.23.880} as
\begin{eqnarray*}
\Delta S&=&S(\alpha_{i_1}|\alpha_{i_2})-\left(
S(\alpha_{i_1}|\beta_{j_2})+S(\beta_{j_2}|\beta_{j_1})+S(\beta_{j_1}|\alpha_{i_2})\right).
\end{eqnarray*}
According to the results of {\bf Theorem \ref{bell}}, $\Delta S$ is bounded by negative values when the complete probability table exists. If the theory is based on local realism, one can always prepare the complete table to observe codewords for both $\B$ and $\CC$. In order to design the experiment that gives a stronger correlation ($\Delta S>0$), one has to employ a rule for choosing the probability table, i.e., a choice that cannot be determined locally. Moreover, the rule must also satisfy requirements from special relativity, if one would like to interpret as physical law. The stochastic double codeword-transfer experiment defined by {\bf Definition \ref{dcodetf3}} is an example of such a rule. Under {\bf Definition \ref{dcodetf3}-\ref{dcodetf3-i4b}}, for instance, $\B$ cannot know the probability table he is using because it depends on $\CC$'s decision, and that cannot be known by $\B$. This lack of the complete probability table is deeply related to the Kochen--Specker theorem (KST). The KST insists of absence of a complete set of physical quantities without measurements in QM, and corresponds exactly to lack of the complete probability table introduced in {\bf Definition \ref{dcodetf3}}. Moreover, if we look at only $\B$'s results, we cannot extract any information about $\CC$'s choices and results; that means $\CC$'s information cannot transferred to $\B$ immediately, which is a requirement from special relativity. This coexistence of nonlocality and special relativity is realized by the rule of {\bf Definition \ref{dcodetf3}-\ref{dcodetf3-i4b}} of the stochastic double codeword-transfer experiment. The probability tables, $\mu_B(\alpha_1)$ and $\mu_B(\alpha_2)$, include $\theta_c$, though these are tables for $\B$, which is called ``{\it entanglement}''. However, $\B$ cannot extract a value of $\theta_c$ because $\theta_c$ appears only in phase of the unitary transformation and disappears after reaching the average. Violation of Bell's inequality can be judged by checking whether the correlation $\Delta S$ is greater than zero or not. Numerical results with employing rule of {\bf Definition \ref{dcodetf3}-\ref{dcodetf3-i4b}} are calculated and shown in Fig. \ref{bell23}. One can clearly see the violation of Bell's inequality in some parameter regions.

This trick can be implemented because the probability table represents probability amplitude. For example, if $\B$ decided to measure a codeword from the same set as $\CC$, say $\V_A$, state vectors for $\B$ is superposition of two possible sates depending on the result of measurement of $\CC$ such as
\begin{eqnarray*}
\left(
    \begin{array}{c}
      \alpha_1 \\
      \alpha_2 
    \end{array}
  \right)
=&~&R(\theta_b)\cdot
\left(
    \begin{array}{c}
      \cos{\theta_\alpha} \\
      \sin{\theta_\alpha}
    \end{array}
  \right)\Bigl|_{\theta_\alpha\rightarrow\theta_c}\times
(1,0)\cdot R(\theta_c)\cdot
\left(
    \begin{array}{c}
      \cos{\theta_\alpha} \\
      \sin{\theta_\alpha}
    \end{array}
  \right)\\
&+&R(\theta_b)\cdot
\left(
    \begin{array}{c}
      \cos{\theta_\alpha} \\
      \sin{\theta_\alpha}
    \end{array}
  \right)\Bigl|_{\theta_\alpha\rightarrow\theta_c+\pi/2}\times
(0,1)\cdot R(\theta_c)\cdot
\left(
    \begin{array}{c}
      \cos{\theta_\alpha} \\
      \sin{\theta_\alpha}
    \end{array}
  \right),\\
=&~&\left(
    \begin{array}{c}
      \cos{(\theta_c-\theta_b)}\cos{(\theta_\alpha-\theta_c)}
     -\sin{(\theta_c-\theta_b)}\sin{(\theta_\alpha-\theta_c)}\\
      \sin{(\theta_c-\theta_b)}\cos{(\theta_\alpha-\theta_c)}
     -\sin{(\theta_c-\theta_b)}\cos{(\theta_\alpha-\theta_c)}\\
     
    \end{array}
  \right),
\end{eqnarray*}
where $R(\theta)$ is a rotation matrix. Then the probabilities of $\mu_{dc}(\alpha_i)$ are obtained as in {\bf Definition \ref{dcodetf3}-\ref{dcodetf3-i4b}}. The nonlocal realism is induced by squaring the state vector after superposition of two possible states. This is another example outlining the advantage of the method using probability amplitude. 
%%%%%%%%%%%%%%%%%%%%%%%%%%%%%%%%%%%%%%%%%%%%%%%%%%%%%%%%%%%%%%%%%%%%%%%%%
\begin{figure}[t]
 \begin{minipage}{0.5\hsize}
  \begin{center} 
     \includegraphics[height=4cm]{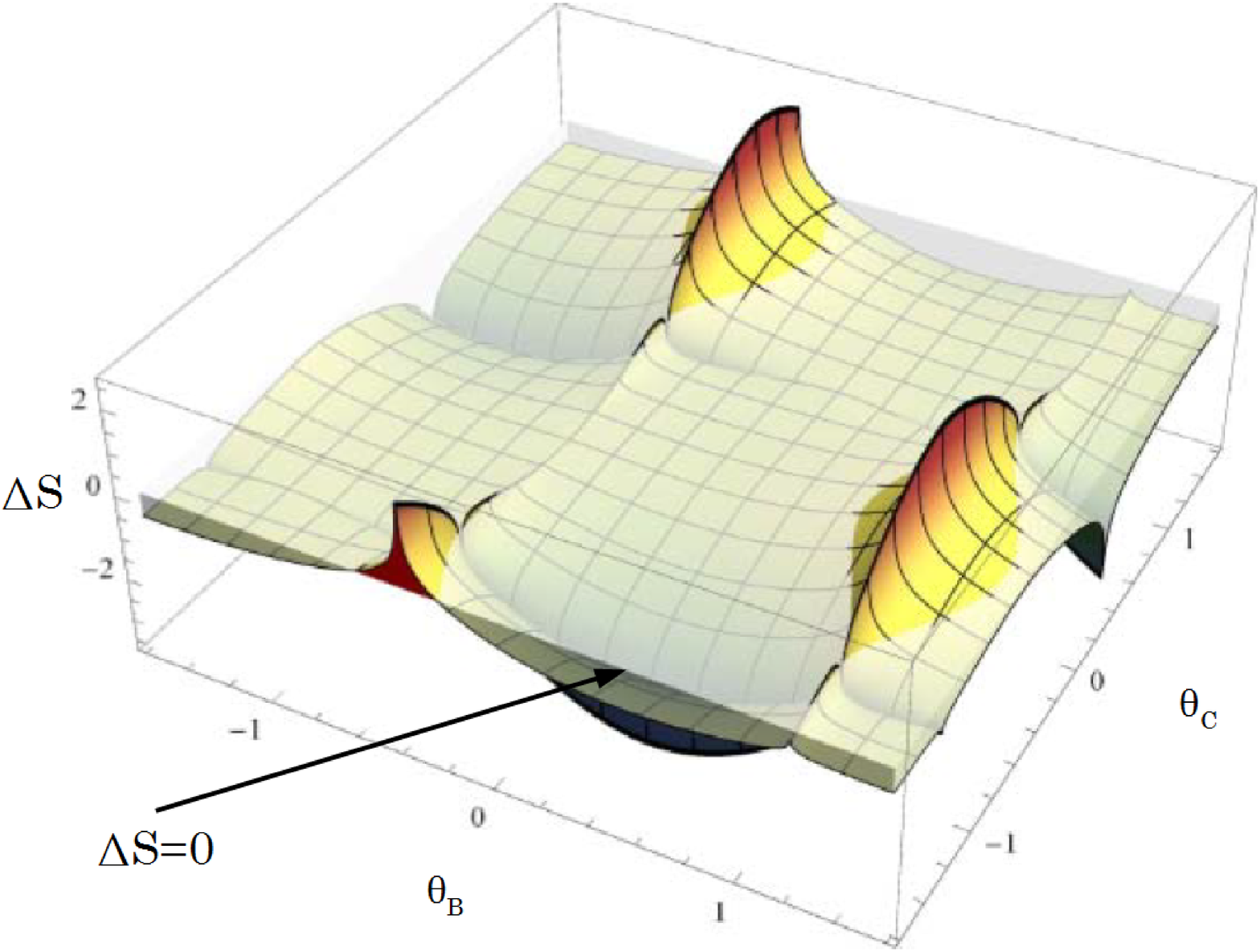}%
   \end{center}
  \end{minipage}
  \begin{minipage}{0.45\hsize}
   \begin{center} 
     \includegraphics[height=4cm]{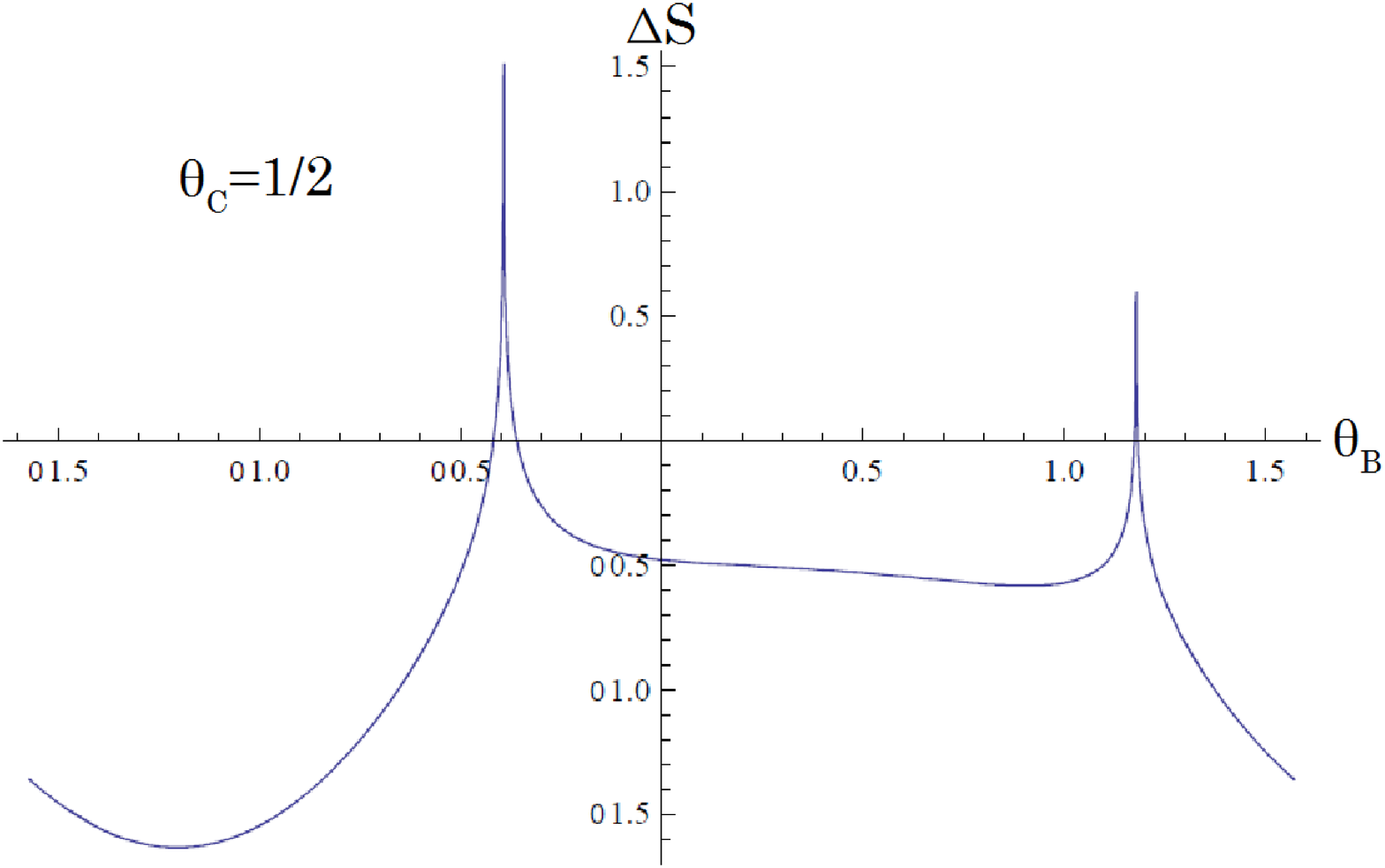}%
   \end{center}
  \end{minipage}
 \begin{center} 
 \caption{$\Delta S$ for the stochastic double codeword-transfer experiment without a complete probability table. It can be seen that Bell's inequality is broken in a part of parameter region.} \label{bell23}
 \end{center} 
\end{figure}
%%%%%%%%%%%%%%%%%%%%%%%%%%%%%%%%%%%%%%%%%%%%%%%%%%%%%%%%%%%%%%%%%%%%%%%%%
%%%%%%%%%%%%%%%%%%%%%%%%%%%%%%
\section{Summary}
In this report, a basic definition of quantum mechanics concerning its static aspect is proposed. Here, a dynamic aspect of quantum mechanics is not treated. The simple codeword-transfer experiment which satisfies the definition of quantum mechanics is designed to investigate some of it aspects. Then it is proved that a method using probability amplitude gives the minimum error for the physical observables using information theory. Also, it is shown that the size of the error doesn't depend on parametrization of the coding. 
Nonlocal realism is one of the most essential parts of the nature of quantum mechanics. It is shown that quantum mechanics defined here can include nonlocal realism for the double codeword-transfer experiment introduced by extending a codeword-transfer experiment, above. We showed that the quantum mechanics defined here can violate Bell's inequality, thanks to the property of the probability amplitude.

In conclusion, the probability amplitude rather than probability density gives the minimum and independent mean-square errors from parametrization. Moreover, it allows one to obtain nontrivial and nonlocal correlation on two independent measurements which violate Bell's inequality incorporate with the Kochen--Specker theorem. It is worth pointing out that nonlocal realism can be realized without any complex-number valued amplitude here. The complex-number valued amplitude could be one of convenient representations for quantum mechanics, but indispensable ingredient of that.
\section*{Acknowledgments}
We wish to thank to Dr. Y.~Sugiyama, Profs. T.~Kaneko, K.~Kato, T.~Kon, and F.~Yuasa for their continuous encouragement and fruitful discussions. We wish to thank Prof. Tsutsui for his excellent lecture about advanced quantum mechanics. 
%We also thank anonymous reviewer of International Journal of Theoretical Physics for his/her critical review. This manuscript may be much improved due to his/her useful comments, even though our submission was rejected by the editor.
%We would like to thank Enago for the English language review.
% BibTeX users please use
\bibliographystyle{plain}	% (uses file "plain.bst")
%\bibliography{ref}

%
% Appendix
%
\newpage
\appendix
\section{Appendix}\label{ap}
\subsection{Classical estimation theory}\label{ap1}
We define terms associated with physical measurement according to classical estimation theory\cite{covert91-12-11} as follows. Let $\X$ be a random variable for a given physical system described by the $N$-tuple 
${\bm \theta}=\{\theta_1,\cdots,\theta_N\}$, where $\theta_i$ is the {\it i~th physical parameter}. The set of all possible values of 
$\theta_i \in \R$, denoted by $\Theta$, is called the {\it parameter set}. The random variable $\X$ is distributed according to the probability density function ${\rm f}(x;{\bm \theta})\geq 0$, which is normalized as $\int_{x\in \Omega} dx~{\rm f}(x;{\bm \theta})=1$, where $x\in\R$ is one possible value of the whole event $(=\Omega)$. 
For physical applications, we introduce the {\it probability amplitude} defined by
\begin{eqnarray*}
\left| 
\omega(x;{\bm \theta})
\right|^2={\rm f}(x;{\bm \theta}).
\end{eqnarray*}
A part of {\it experimental apparatus} is assumed to output numbers distributed according to the probability density. Any resulting set of numbers $\X_n=\{x_1,\cdots,x_n \}$, drawn independently and identically distributed (i.i.d.), is called the {\it experimental data}. The estimate of the physical parameter is called a {\it measurement}. Because experimental data are i.i.d., the corresponding probability density function can be expressed as a product:
\begin{eqnarray*}
{\rm f}(\X_n;{\bm \theta})=\prod_{j=1}^n~{\rm f}(x_j;{\bm \theta}).
\end{eqnarray*}
A function mapping the experimental data to one possible value of the parameter set such as 
\begin{eqnarray*}
T_i:\X_n \rightarrow \Theta:\{x_1,\cdots,x_n\}\mapsto{\tilde \theta}_i
\end{eqnarray*}
is called an {\it estimator} for the $i$th physical parameter, denoted by
$
T_i(\X_n)={\tilde \theta_i}.
$
The {\it experimental error} in the $i$th physical parameter is defined as the root mean square error: 
\begin{eqnarray*}
\epsilon_i = E[(T_i(\X_n)-\theta_i)^2]^{1/2},
\end{eqnarray*}
where $\theta_i$ is the true value of the $i$ th physical parameter. True values of physical parameters are typically unknown, but a mean-square error can be reduced below any desired value by accumulating a sufficiently large amount of experimental data, thanks to the law of large numbers. If the mean value of the experimental error converges to zero in probability, i.e.,
\begin{eqnarray*}
\lim_{n\rightarrow\infty}E_{\theta_i}[\tilde\theta_i-\theta_i]\rightarrow 0~~
\rm{(in~probability)},
\end{eqnarray*} 
after accumulation of infinitely many statistics,
that estimator is called an {\it unbiased estimator}. Among such estimators, the one giving the least error is called the {\it best estimator}.
\subsection{Information theory}\label{ap2}
For a probability space $(\Omega,{\cal A},P)$ and probability variable $X$ defined on it, information entropy $S(X)$ is defined as
\begin{eqnarray*}
S(X)=-\sum_{x \in \Omega}P(x) \log{P(x)}.
\end{eqnarray*}
$S(X)\geq0$ immediately follows from $0\leq P\leq1$.
For two probability variable $X,Y$ whose domains are $\Omega_x,\Omega_y$, where $\Omega_x,\Omega_y \subseteq \Omega$, a joint entropy is defined as 
\begin{eqnarray*}
S(X,Y)=-\sum_{x \in \Omega_x}\sum_{y \in \Omega_y}P(x\cap y) \log{P(x\cap y)},
\end{eqnarray*}
where $P(x\cap y)$ is a probability to observe $x$ in $X$ and $y$ in$Y$, simultaneously. A conditional entropy is defined as
\begin{eqnarray*}
S(Y|X)=-\sum_{x \in \Omega_x}\sum_{y \in \Omega_y}P(x\cap y) \log{P(y|x)},
\end{eqnarray*}
Where $P(y|x)$ is conditional probability to observe $y$ in $Y$ when $x$ in$X$ is obtained. On those entropies, following formulae are obtained:
\begin{eqnarray*}
S(X,Y)&=&S(X)+S(Y|X),\\
S(X|Y)&\leq&S(X).
\end{eqnarray*}
\end{document}